\documentclass[reqno]{amsart}

\usepackage{amsfonts}
\usepackage{amsmath}
\usepackage{amssymb}

\usepackage{graphicx}
\usepackage{subfig}
\usepackage[update,prepend]{epstopdf}
\usepackage{url}

\newtheorem{theorem}{Theorem}
\theoremstyle{plain}

\AtBeginDocument{{\noindent\small
This is a preprint of a paper whose final and definite form is with 
\emph{Commun. Fac. Sci. Univ. Ank. Ser. A1 Math. Stat.}, ISSN: 1303-5991. 
Submitted 27 Sept 2016; Article revised 05 Apr 2017;
Article accepted for publication 01 May 2017.} 
\vspace{9mm}}

\begin{document}

\title[A SEIR model for Ebola virus with demographic effects]{Analysis, simulation 
and optimal control\\ of a SEIR model for Ebola virus\\ with demographic effects}


\author{Amira Rachah}
\address{Amira Rachah: Department of Production Animal Clinical Sciences,
Norwegian University of Life Sciences, PO Box 8146, NO-0033 Oslo, Norway.}
\email{amira.rachah@nmbu.no}


\author{Delfim F. M. Torres}
\address{Delfim F. M. Torres: Center for Research 
and Development in Mathematics and Applications (CIDMA),
Department of Mathematics, University of Aveiro, 3810-193 Aveiro, Portugal.}
\email{delfim@ua.pt}


\date{Received: September 27, 2016, Revised: April 05, 2017; Accepted: May 01, 2017.}




\begin{abstract}
Ebola virus is one of the most virulent pathogens for humans.
We present a mathematical description of different 
Susceptible--Exposed--Infectious--Recovered (SEIR) models.
By using mathematical modeling and analysis, the latest major 
outbreak of Ebola virus in West Africa is described. Our aim 
is to study and discuss the properties of SEIR models with 
respect to Ebola virus, the information they provide, and 
when the models make sense. We added to the basic SEIR model 
demographic effects in order to analyze the equilibria with 
vital dynamics. Numerical simulations confirm the theoretical 
analysis. The control of the propagation of the virus through 
vaccination is investigated and the case study 
of Liberia is discussed in detail.
\end{abstract}


\maketitle


\section{Introduction}

Ebola is a deadly virus that attacks healthy cells and replicates itself
in a host's body. The virus, previously known as Ebola hemorrhagic fever,
is the deadliest pathogen for humans and has recently affected
several African countries. Discovered in 1976 in Central Africa, the recent
outbreaks affected the more heavily populated countries of West Africa
\cite{barraya,tara1}. Early symptoms of Ebola include: fever, headache, joint
and muscle aches, sore throat, and weakness. Later symptoms include diarrhea,
vomiting, stomach pain, hiccups, rashes, bleeding, and organ failure.
When Ebola progresses to external and internal bleeding, it is almost
always fatal \cite{alton,legrand,anon2,okwar,anon1}.
Ebola virus is transmitted initially to human by contact with an infected
animal's body fluid. Ebola is most commonly spread by contact with blood
and secretions, either via direct contact (through broken skin or mucous
membranes in, e.g., the eyes, nose, or mouth) with the infected
individual or fluids on clothing or other surfaces, as well as needles
\cite{borio,edward,dowel,peter,tara2}.

Epidemic models date back to the early twentieth century,
to the 1927 work by Kermack and McKendrick, whose models
were used to study the plague and cholera epidemics
\cite{Kermack:McKendrick:I,Kermack:McKendrick:II}.
Epidemic modeling is nowadays a powerful tool for investigating human infectious
diseases, such as Ebola, contributing to the understanding of the dynamics
of virus, providing useful predictions about the potential transmission
of the virus and the effectiveness of possible control measures,
which can provide valuable information for public health policy makers
\cite{diekman,gaff,kretz,longini,MR2719552,delf}.

The most commonly implemented models in epidemiology
are the SIR and SEIR models. The SIR model consists of three compartments:
Susceptible individuals $S$, Infectious individuals $I$,
and Recovered individuals $R$. In many infectious diseases
there is an exposed period after the transmission of the infection
from susceptible to potentially infective members, but before these potential
infective can transmit infection. Then an extra compartment is introduced,
the so called exposed class $E$, and we use compartments $S$, $E$, $I$ and $R$
to give a generalization of the basic SIR model \cite{brauer}.
When analyzing a new outbreak, researchers usually start with the 
basic SIR and SEIR models to fit the available outbreak data, obtaining 
estimates for the parameters of the model. Only after that, more complicated
models may be considered \cite{brauer}. In case of Ebola, the SIR model 
has already been deeply explored in the literature \cite{MyID:321,MyID:336}.
For a case study of the Ebola virus in Guinea, through a SIR model, 
we refer to \cite{MyID:344}. The results obtained by SIR models
are good enough, taking into account their simplicity. However, 
the transmission of Ebola virus is better described by a SEIR model. 
This is because it takes a certain time for an infected individual 
to become infectious. During that period of time, such individuals 
are in the exposed/latent compartment. A mathematical description 
of the spread of Ebola virus based on the basic SEIR model 
has been carried out, e.g. in \cite{SEIR:Mamo:Koya,MyID:340}. Discrete SEIR
time models to Ebola epidemics are available in \cite{Ebola:Bartlett}. 
In our work, we are interested in continuous time models, 
which are more common with respect to Ebola modeling  
\cite{Ebola:Atangana,Ebola:Boujakjian,SEIR:Mamo:Koya}.
In \cite{Ebola:Atangana}, the homotopy decomposition
method is used to solve a system of equations
modeling Ebola hemorrhagic fever involving the
so called beta derivative, which can be 
considered as the fractional order of the system. 
In our case, we deal with classical derivatives 
and standard integer-order systems. For fractional
modeling of Ebola, see also \cite{Area1}.
Our main control measure for the propagation of the virus is
vaccination. For the use of quarantine as a control measure
to limit the transmission of the Ebola virus using a SEIR model,
we refer the interested reader to \cite{Ebola:Boujakjian}. 
For a comparison study between the basic SIR and SEIR models 
to describe an Ebola outbreak, and for the conclusion 
of superiority of the SEIR model, we refer the reader to \cite{MyID:331}. 
It turns out that in available Ebola studies with SEIR models, 
the population is assumed to be constant: 
see Remark~1 of \cite{MyID:340} and Proposition~2.2 
of \cite{MyID:331}. This assumption is far from being true 
in West African countries. For example, in Liberia, the birth rate 
is approximately four times the death rate \cite{indexmundi}. 
Motivated by this fact, here our main aim is to study the latest 
major outbreak of Ebola virus occurred in Liberia 
through an appropriate SEIR model with vital dynamics, 
which takes into account the demographic effects on the population.
This is in contrast with all available results in the literature.

The paper is organized as follows. In Section~\ref{sec:2}, we recall
the basic mathematical SEIR model to describe the dynamics of the Ebola virus
that recently affected West Africa. After the modeling,
we analyze mathematically the SEIR model. In Section~\ref{sec:3},
we propose and analyze a new SEIR model with vital dynamics 
by adding demographic effects (Section~\ref{sec:3.1}). 
The equilibria of the model are studied in Section~\ref{subsec:24}.
Then, a numerical simulation is presented, which 
confirms the theoretical analysis (Section~\ref{sec:Sim:SEIR:de}).
After numerical resolution of the model, in Section~\ref{sec:4} 
we control the propagation of the virus through vaccination, 
reducing the number of infected individuals while taking into account 
the cost of vaccination. Finally, in Section~\ref{sec:5}, we propose
and investigate a more general model with demographic effects, for which 
there is an increase of the death rates for the exposed and infectious classes.
We show that such model describes well the outbreak of Ebola virus occurred 
in Liberia in 2014 \cite{WHO:Lib:2014}. We end with Section~\ref{sec:conc} 
of conclusions.


\section{Formulation of the basic SEIR model}
\label{sec:2}

In this section, we briefly recall the analysis of the properties 
of the basic SEIR system of equations that has been used to 
describe the recent outbreak of Ebola virus in West Africa
\cite{MyID:340,MyID:331}. The description of the transmission
of Ebola virus by the SEIR model is based on the subdivision
of the population into four compartments:
\begin{itemize}

\item Susceptible compartment $S(t)$, which denotes individuals who are
susceptible to catch the virus, and so might become infectious if exposed.

\item Exposed compartment $E(t)$, which denotes the individuals who are infected
but the symptoms of the virus are not yet visible.

\item Infectious compartment $I(t)$, which denotes infectious individuals
who are suffering the symptoms of Ebola and are able to spread the virus
through contact with the susceptible class of individuals.

\item Recovered compartment $R(t)$, which denotes individuals who have immunity
to the infection and, consequently, do not affect the transmission dynamics,
in any way, when in contact with other individuals.
\end{itemize}
The SEIR model is an extension of the simpler SIR model 
\cite{SEIR:Mamo:Koya,MyID:321,MyID:336,MyID:344}.
The particularity of the SEIR model is in the exposed compartment,
which is characterized by infected individuals that cannot communicate
yet the virus. These individuals are in the so called latent period \cite{brauer}.
For Ebola virus, this stage makes all sense since
it takes a certain time for a susceptible individual at time $t$,
denoted by $S(t)$, to enter the Infectious compartment $I(t)$.
Because the recovered individuals $R(t)$ have immunity to the infection,
they do not affect the transmission dynamics in any way when in contact with
other individuals. Figure~\ref{fig1:SEIR} shows the diagrammatic
representation of virus progress in an individual.
\begin{figure}[ht]
\centering
\includegraphics[width=12cm]{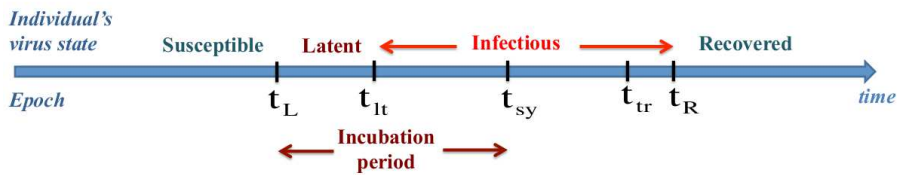}
\caption{Ebola virus progress in an individual by using the SEIR model,
where infectious occurs at $t_L$, latency to infectious transition
at $t_{lt}$, symptoms appear at $t_{sy}$, first transmission to another
susceptible at $t_{tr}$, and individual is no longer infectious
(recovered) at $t_R$. \label{fig1:SEIR}}
\end{figure}
The transmission of the virus is then described by the following
system of nonlinear ordinary differential equations:
\begin{equation}
\label{eq1:SEIR}
\begin{cases}
\dfrac{dS(t)}{dt} = -\beta S(t)I(t),\\[0.25cm]
\dfrac{dE(t)}{dt} = \beta S(t)I(t) - \gamma E(t),\\[0.25cm]
\dfrac{dI(t)}{dt} = \gamma E(t) - \mu I(t),\\[0.25cm]
\dfrac{dR(t)}{dt} = \mu I(t),
\end{cases}
\end{equation}
where $\beta \geq 0$ is the transmission rate; $\gamma \geq 0$ is the infectious
rate; and $\mu \geq 0$ is the recovery rate.
The initial conditions are given:
\begin{equation*}
S(0)=S_0>0,
\quad E(0)=E_0 \geq 0,
\quad I(0)=I_0>0,
\quad R(0)=0.
\end{equation*}
From \eqref{eq1:SEIR}, we see that
$\dfrac{d}{dt} \left[ S(t) + E(t) + I(t) + R(t) \right] = 0$, that is,
the population $N$ is constant along time:
\begin{equation*}
S(t) + E(t) + I(t) + R(t) =  N
\end{equation*}
for any $t \geq 0$.


\section{SEIR model with demographic effects}
\label{sec:3}

In the well-known basic SEIR model of Section~\ref{sec:2},
one ignores the demographic effects on the population.
In this section, we study a model with vital dynamics
by considering the birth and death rates.
Such model is new in the Ebola context
\cite{Area1,MyID:321,MyID:336,MyID:344,MyID:340,MyID:331,MR3356525}.


\subsection{Model formulation}
\label{sec:3.1}

We expand the SEIR model by including demographic effects: we assume
a constant birth rate $\delta$ and a natural death rate $\lambda$, obtaining
\begin{equation}
\label{eq1:SEIR_BD}
\begin{cases}
\dfrac{dS(t)}{dt} = \delta N  -\beta S(t)I(t) -\lambda S(t),\\[0.25cm]
\dfrac{dE(t)}{dt} = \beta S(t)I(t) - \gamma E(t) - \lambda E(t),\\[0.25cm]
\dfrac{dI(t)}{dt} = \gamma E(t) - \mu I(t) - \lambda I(t),\\[0.25cm]
\dfrac{dR(t)}{dt} = \mu I(t) - \lambda R(t).
\end{cases}
\end{equation}
Figure~\ref{fig2:SEIR} shows the relationship between the variables
of system \eqref{eq1:SEIR_BD}, which describes the SEIR model with
vital dynamics, that is, with demographic effects (birth and death).
\begin{figure}[ht]
\centering
\includegraphics[scale=0.4]{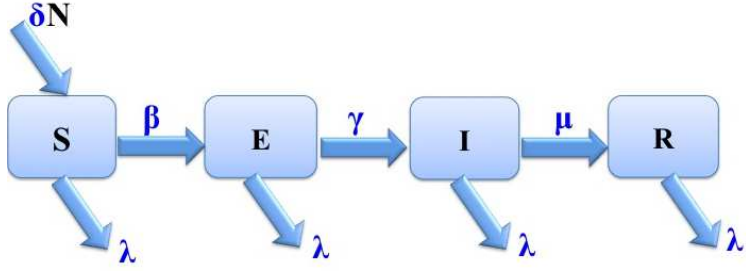}
\caption{Compartment diagram of the SEIR model
\eqref{eq1:SEIR_BD} with vital dynamics.
\label{fig2:SEIR}}
\end{figure}


\subsection{Analysis of the equilibria}
\label{subsec:24}

Let us find the equilibria points of the system of equations
\eqref{eq1:SEIR_BD} that describes the model. By setting
the right-hand side of \eqref{eq1:SEIR_BD} to zero, we get
\begin{equation}
\label{eq1:SEIR_BDeq}
\delta N  -\beta SI -\lambda S = 0,
\end{equation}
\begin{equation}
\label{eq2:SEIR_BDeq}
\beta SI - \gamma E  -\lambda E = 0,
\end{equation}
\begin{equation}
\label{eq3:SEIR_BDeq}
\gamma E - \mu I  -\lambda I = 0,
\end{equation}
\begin{equation}
\label{eqR:SEIR_BDeq}
\mu I  -\lambda R = 0.
\end{equation}
By adding \eqref{eq1:SEIR_BDeq} and \eqref{eq2:SEIR_BDeq}, we obtain that
$\delta N  -\lambda S - \left(\gamma + \lambda\right)E = 0$. Then,
\begin{align}
\label{eq4:SEIR_BDeq}
S = \dfrac{ \delta N -\left(\gamma + \lambda\right)E}{\lambda}.
\end{align}
From \eqref{eq3:SEIR_BDeq} we obtain that
\begin{align}
\label{eq5:SEIR_BDeq}
I = \dfrac{\gamma E}{\mu + \lambda}\, ,
\end{align}
while from \eqref{eqR:SEIR_BDeq} it follows that
\begin{align}
\label{eqR1:SEIR_BDeq}
R = \dfrac{\mu I}{\lambda}.
\end{align}
From \eqref{eq2:SEIR_BDeq}, \eqref{eq4:SEIR_BDeq} and \eqref{eq5:SEIR_BDeq},
we get
\begin{align*}
E \left( \frac{\beta \gamma \left(\delta N
- (\gamma +\lambda) E \right)}{\lambda\left(\mu +\lambda\right)}
- \left(\gamma +\lambda\right)\right)  = 0.
\end{align*}
Therefore, or $E=0$ or
\begin{equation}
\label{eq:Edif0}
E = \dfrac{\beta\gamma\delta N
- \left(\gamma + \lambda\right)\left(\mu + \lambda\right)\lambda}
{\beta \left(\gamma + \lambda\right)\gamma}.
\end{equation}
For $E=0$, from \eqref{eq4:SEIR_BDeq} we obtain that $S=\dfrac{\delta N}{\lambda}$
while from \eqref{eq5:SEIR_BDeq} we get $I=0$. It follows from \eqref{eqR1:SEIR_BDeq}
that $R = 0$. We just proved that there is a virus free equilibrium given by
$P_0=\left(\dfrac{\delta N}{\lambda},0,0,0\right)$.
From \eqref{eq:Edif0} we know that there is another equilibrium with
\begin{align}
\label{eq7:SEIR_BDeq}
E^* =  \dfrac{\delta N }{\gamma + \lambda}
- \dfrac{\lambda\left(\mu +\lambda\right)}{\beta \gamma}.
\end{align}
By using $E$ given by \eqref{eq7:SEIR_BDeq} in \eqref{eq4:SEIR_BDeq}, we get
\begin{align}
\label{eq8:SEIR_BDeq}
S^*=\dfrac{\left(\gamma +\lambda\right)\left(\mu+\lambda\right)}{\beta \gamma},
\end{align}
by substituting \eqref{eq7:SEIR_BDeq} in \eqref{eq5:SEIR_BDeq}, we obtain that
\begin{align}
\label{eq9:SEIR_BDeq}
I^* = \dfrac{\gamma \delta N }{\left(\mu +\lambda\right)\left(\gamma +\lambda\right)}
 - \dfrac{\lambda}{\beta},
\end{align}
and finally using \eqref{eq7:SEIR_BDeq} in \eqref{eqR1:SEIR_BDeq} we get
\begin{align}
\label{eq:SEIR_BDeq:R*}
R^* = \dfrac{\mu}{\lambda}I^* = \dfrac{\mu}{\lambda}\left[
\dfrac{\gamma \delta N }{\left(\mu +\lambda\right)\left(\gamma +\lambda\right)}
 - \dfrac{\lambda}{\beta}
 \right].
\end{align}
We just obtained the second equilibrium point $P^*=\left(S^*,E^*,I^*,R^*\right)$ given by expressions
\eqref{eq7:SEIR_BDeq}--\eqref{eq:SEIR_BDeq:R*}. 

\begin{theorem}
\label{thm1}
Let $S(t)$, $E(t)$, $I(t)$, $R(t)$  
be a solution of the SEIR model \eqref{eq1:SEIR_BD}. 
Then the basic reproduction ratio is given by 
 \begin{align}
\label{eq:R0}
R_0 := \dfrac{\beta \gamma \delta N}{\left(\mu 
+\lambda\right)\left(\gamma +\lambda\right)\lambda}.
\end{align}
\begin{itemize}
\item If $R_0>1$, then the equilibrium $P^*=\left(S^*,E^*,I^*,R^*\right)$
of the virus is obtained, in agreement with expressions
\eqref{eq7:SEIR_BDeq}--\eqref{eq:SEIR_BDeq:R*},
and the virus is able to invade the population.

\item If $R_0<1$, then the disease free equilibrium
$P_0=\left(\dfrac{\delta N}{\lambda},0,0,0\right)$ of the virus is obtained,
which corresponds to the case when the virus dies out (no epidemic).
\end{itemize}
\end{theorem}

\begin{proof}
For computing the basic reproduction ratio $R_0$, 
we apply the next generation method \cite{DiekmannR0,HeffernanR0}.  
Assume that there are $n$ infective classes in the model and 
define the vector $\bar{x}=x_i$, where $x_i$, $i=1,2,\dots,n$, 
denotes the number or the proportion of individuals in the $i$th infective class. 
Let $F_i(\bar{x})$ be the rate of appearance of new infections in the 
$i$th class and let $V_i(\bar{x})=V_i^{-}(\bar{x})-V_i^{+}(\bar{x})$,  
where $V_i^{+}$ consists of transfer of individuals into class $i$
and $V_i^{^-}$ consists of transfer of individuals out of class $i$. 
The difference $F_i(\bar{x})-V_i(\bar{x})$  gives the rate of change 
of $x_i$. Notice that $F_i$ consists of new infections from susceptible,  
whereas $V_i$ includes the transfer of infected individuals 
from one infected class to another \cite{HeffernanR0}.
We can then form the next generation matrix 
from the partial derivatives of $F_i$ and $V_i$:
\begin{align*}
F=\left[ \dfrac{\partial F_i(x_0)}{\partial x_j}\right],  
\quad V=\left[ \dfrac{\partial V_i(x_0)}{\partial x_j}\right], 
\end{align*}
where $i,j=1,2,\dots,n$ and $x_0$ is the initial condition of the epidemic.
The basic reproduction ratio  $R_0$ is given  by  the dominant eigenvalue 
of the matrix $FV^{-1}$ \cite{HeffernanR0}. Applying
the next generation method to the SEIR model \eqref{eq1:SEIR_BD},  
and since we are only concerned with individuals that spread 
the infection, we only need to model the exposed, $E$, and infected, $I$, classes. 
Let us define the model dynamics using the equations
\begin{equation*}
\begin{cases}
\dfrac{dE(t)}{dt} = \beta S(t)I(t) - \left(\gamma + \lambda\right) E(t),\\[0.25cm]
\dfrac{dI(t)}{dt} = \gamma E(t) - \left(\mu + \lambda \right) I(t).
\end{cases}
\end{equation*}
For this system, 
\[
F = \left( 
\begin{array}{cc}
0 & \dfrac{\beta N \delta}{\lambda} \\ [0.25cm]
0 & 0 
\end{array} \right),
\] 
where $\delta$ is the birth rate and $\lambda$ is the death rate, and 
\[
V =\left( 
\begin{array}{cc}
\gamma + \lambda & 0 \\ [0.25cm]
-\gamma & \mu + \lambda \end{array} \right).
\] 
Then, 
\[
FV^{-1} = \left( 
\begin{array}{cc}
\dfrac{\beta N \delta \gamma}{\left(\gamma + \lambda\right)\left(\mu 
+ \lambda \right)\lambda} & \dfrac{\beta N \delta }{\left(\mu 
+ \lambda \right)\lambda}  \\ [0.35cm]
0 & 0 \end{array} \right).
\] 
The dominant eigenvalue $R_0$ of $FV^{-1}$ 
is given by expression \eqref{eq:R0}. 
\end{proof}

The recent status of the Ebola virus corresponds to the case $R_0>1$ \cite{who},
which we study numerically in Section~\ref{sec:Sim:SEIR:de}.


\subsection{Simulation of the SEIR model with demographic effects}
\label{sec:Sim:SEIR:de}

Now we solve numerically the SEIR model with vital dynamics \eqref{eq1:SEIR_BD}
by using the parameters presented in the work of Rachah and Torres
\cite{MyID:321,MyID:340}. The early detection of Ebola virus in West Africa
is characterised by $R_0 = 1.95$. Then the Ebola virus is really
an epidemic, invading the populations of West Africa. The parameters
$\beta=0.2$, $\gamma=0.1887$ and $\mu=0.1$, studied by Rachah and Torres
\cite{MyID:321,MyID:340}, are based on the fact that 88\% of population is susceptible,
7\% of population is exposed (infected but not infectious), and 5\% of population
is infectious. In agreement, the initial susceptible, exposed, infectious
and recovered populations, are given respectively by
\begin{equation}
\label{eq:ic}
S(0)=0.88, \quad E(0)=0.07, \quad I(0)=0.05, \quad R(0)=0. 
\end{equation}
In the numerical resolution of the model, 
we take the birth rate $\delta=0.03507$
and the death rate $\lambda=0.0099$ \cite{brauer,indexmundi},
with the other parameters and initial conditions as introduced
by Rachah and Torres \cite{MyID:321,MyID:340}. The birth and death rates
are obtained from the data of population of Liberia in 2014 \cite{indexmundi}.

Figure~\ref{BD_fig1:4} shows the evolution of individuals
over time. We see that the oscillations in the numbers of
compartments $S$, $E$ and $I$ damp out over time,
eventually reaching an equilibrium, respectively $S^*$,
$E^*$ and $I^*$. When we calculate the value of the 
theoretical result \eqref{eq8:SEIR_BDeq}, we find that $S^*=0.57$,
which is equal to the $S^*$ computed by the numerical resolution of the model:
see Figure~\ref{BD_fig1:SEIR}. Figure~\ref{BD_fig2:SEIR} shows the evolution
of the exposed individuals $E(t)$ over time. Note that
the equilibrium $E^*$ is given by \eqref{eq7:SEIR_BDeq}.
When we calculate the value of this theoretical result, we find $E^*=0.14$,
which is equal to the $E^*$ computed by the numerical resolution of the model.
Figure~\ref{BD_fig3:SEIR} shows the evolution of the infected
individuals $I(t)$ over time. In this case the equilibrium
$I^*$ is computed theoretically by \eqref{eq9:SEIR_BDeq},
which agrees with $I^* = 0.25$ computed by the numerical resolution of the model.
Finally, Figure~\ref{BD_fig4:SEIR} shows the evolution of the recovered individuals
$R(t)$ over time. Similarly as before, we found that the equilibrium $R^* = 2.51$,
computed theoretically by \eqref{eq:SEIR_BDeq:R*}, coincides with the numerical resolution
of the model. The fact that the reached equilibrium $(S^*,  E^*, I^*, R^*)$,
computed theoretically, is equal to the value found by the numerical simulation,
is a validation of our study of the SEIR model with vital dynamics.
\begin{figure}[ht]
\centering
\subfloat[Evolution of $S(t)$ with $S^*=0.57$]{
\label{BD_fig1:SEIR}\includegraphics[width=0.5\textwidth]{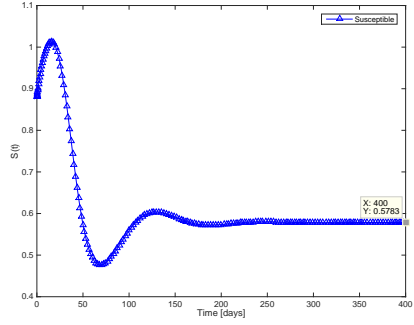}}
\subfloat[Evolution of $E(t)$ with $E^*=0.14$]{
\label{BD_fig2:SEIR}\includegraphics[width=0.5\textwidth]{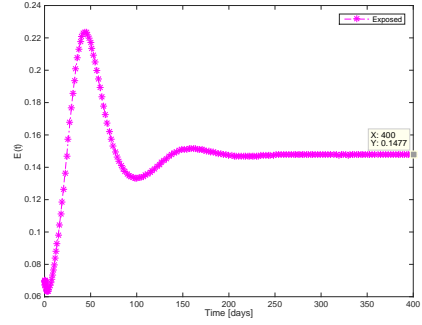}}\\
\subfloat[Evolution of $I(t)$ with $I^*=0.25$]{
\label{BD_fig3:SEIR}\includegraphics[width=0.5\textwidth]{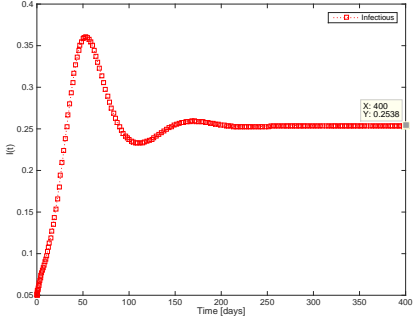}}
\subfloat[Evolution of $R(t)$ with $R^*=2.51$]{
\label{BD_fig4:SEIR}\includegraphics[width=0.5\textwidth]{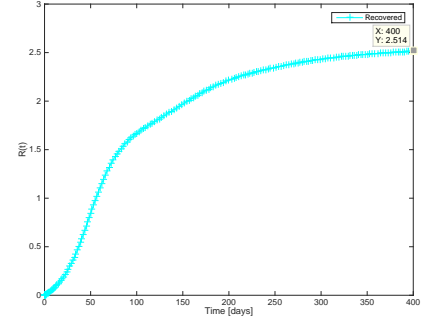}}
\caption{Evolution of individuals in compartments
$S(t)$, $E(t)$, $I(t)$, $R(t)$ of the SEIR model
\eqref{eq1:SEIR_BD} with vital dynamics, where the
endemic equilibrium is given by
$\left(S^*(t), E^*(t), I^*(t), R^*(t)\right)
=\left(0.57, 0.14, 0.25, 2.51\right)$.}
\label{BD_fig1:4}
\end{figure}

\section{Control of the virus with demographic effects through vaccination}
\label{sec:4}

According to recent news, a vaccine against Ebola virus
is ongoing mass production, to be used in countries
like Guinea, Sierra Leone and Liberia \cite{news:china,Reuters}.
Motivated by this fact, we now present
a strategy for the control of the virus
by introducing into the model \eqref{eq1:SEIR_BD}
a control $u(t)$ representing the vaccination rate at time $t$.
Precisely, the control $u(t)$ is the fraction of susceptible individuals
being vaccinated at time $t$. Then, the mathematical model with control
is given by the system of nonlinear differential equations
\begin{equation}
\label{SEIR_control}
\begin{cases}
\dfrac{dS(t)}{dt} = \delta N -\beta S(t)I(t)-\lambda S(t) - u(t) S(t),\\[0.25cm]
\dfrac{dE(t)}{dt} = \beta S(t)I(t) - \gamma E(t) - \lambda E(t),\\[0.25cm]
\dfrac{dI(t)}{dt} = \gamma E(t) - \mu I(t) -\lambda I(t),\\[0.25cm]
\dfrac{dR(t)}{dt} = \mu I(t) - \lambda R(t) + u(t) S(t).
\end{cases}
\end{equation}
The goal of the adopted strategy is to reduce the infected individuals and the 
cost of vaccination on a fixed time interval. Precisely, the optimal control 
problem consists of minimizing the objective functional $J$,
\begin{equation}
\label{cost_func_strat1}
J(I,u) = \int_{0}^{t_{end}} \left[I(t) + \dfrac{\tau}{2}u^2(t)\right] dt
\longrightarrow \min,
\end{equation}
where $u(t)$ is the control variable, which represents the vaccination rate
at time $t$, and the parameters $\tau$ and $t_{end}$ denote, respectively,
the weight on the cost of vaccination and the duration, in days,  
of the vaccination program.
In our study of the control of the virus, we use the parameters
of Section~\ref{sec:Sim:SEIR:de} with $\tau = 0.02$ and $t_{end} = 90$.
For the numerical solution of the optimal control
problem, we have used the \textsf{ACADO} solver \cite{acado},
which is based on a multiple shooting method, including automatic
differentiation and based ultimately on the semidirect multiple shooting
algorithm of Bock and Pitt \cite{acado2}. The \textsf{ACADO} solver comes
as a self-contained public domain software environment, written in \textsf{C++},
for automatic control and dynamic optimization \cite{acado}.
\begin{figure}[ht]
\centering
\subfloat[Susceptible individuals $S(t)$]{
\label{cntrlBD_S}\includegraphics[width=0.5\textwidth]{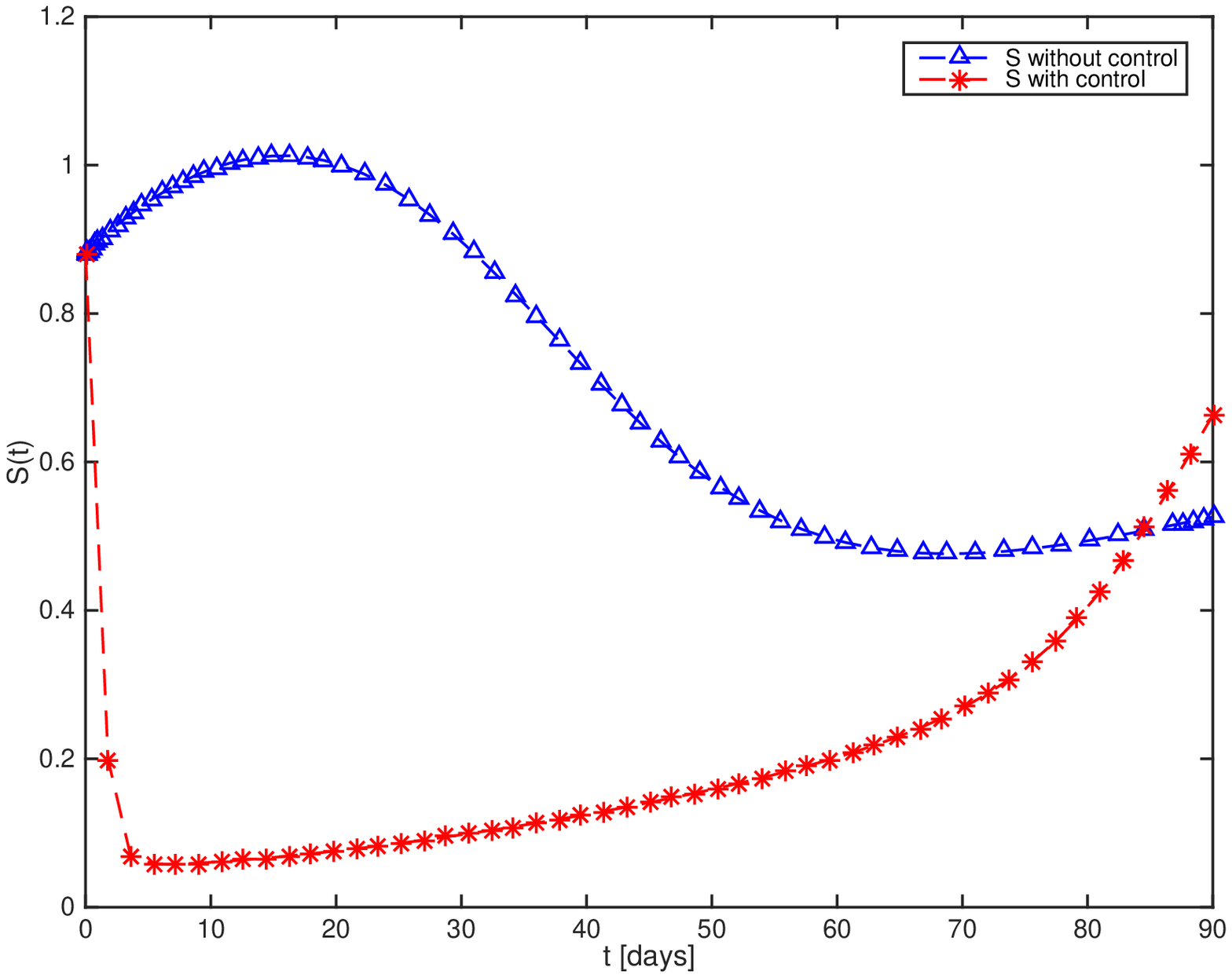}}
\subfloat[Exposed individuals $E(t)$]{
\label{cntrlBD_E}\includegraphics[width=0.5\textwidth]{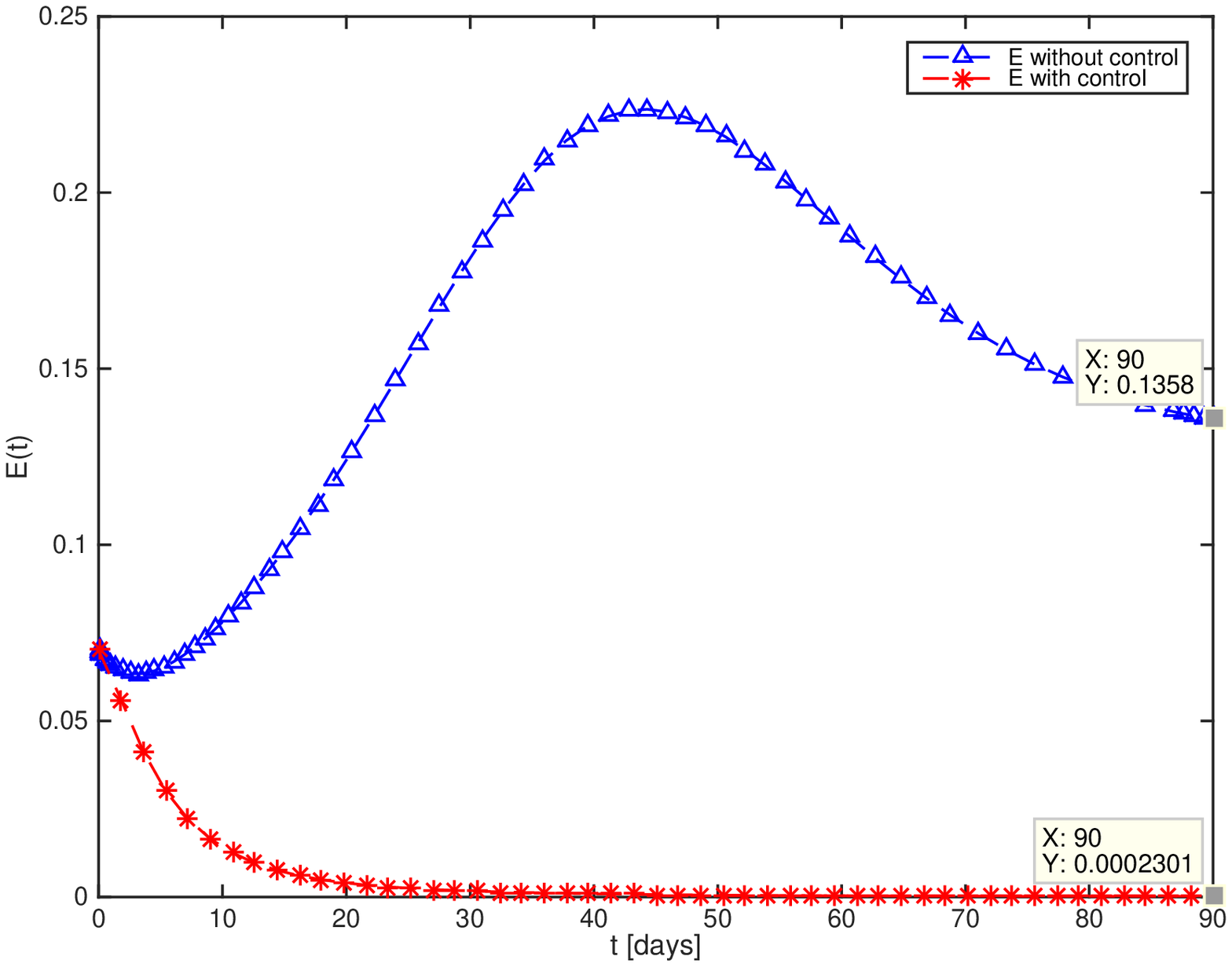}}\\
\subfloat[Infected individuals $I(t)$]{
\label{cntrlBD_I}\includegraphics[width=0.5\textwidth]{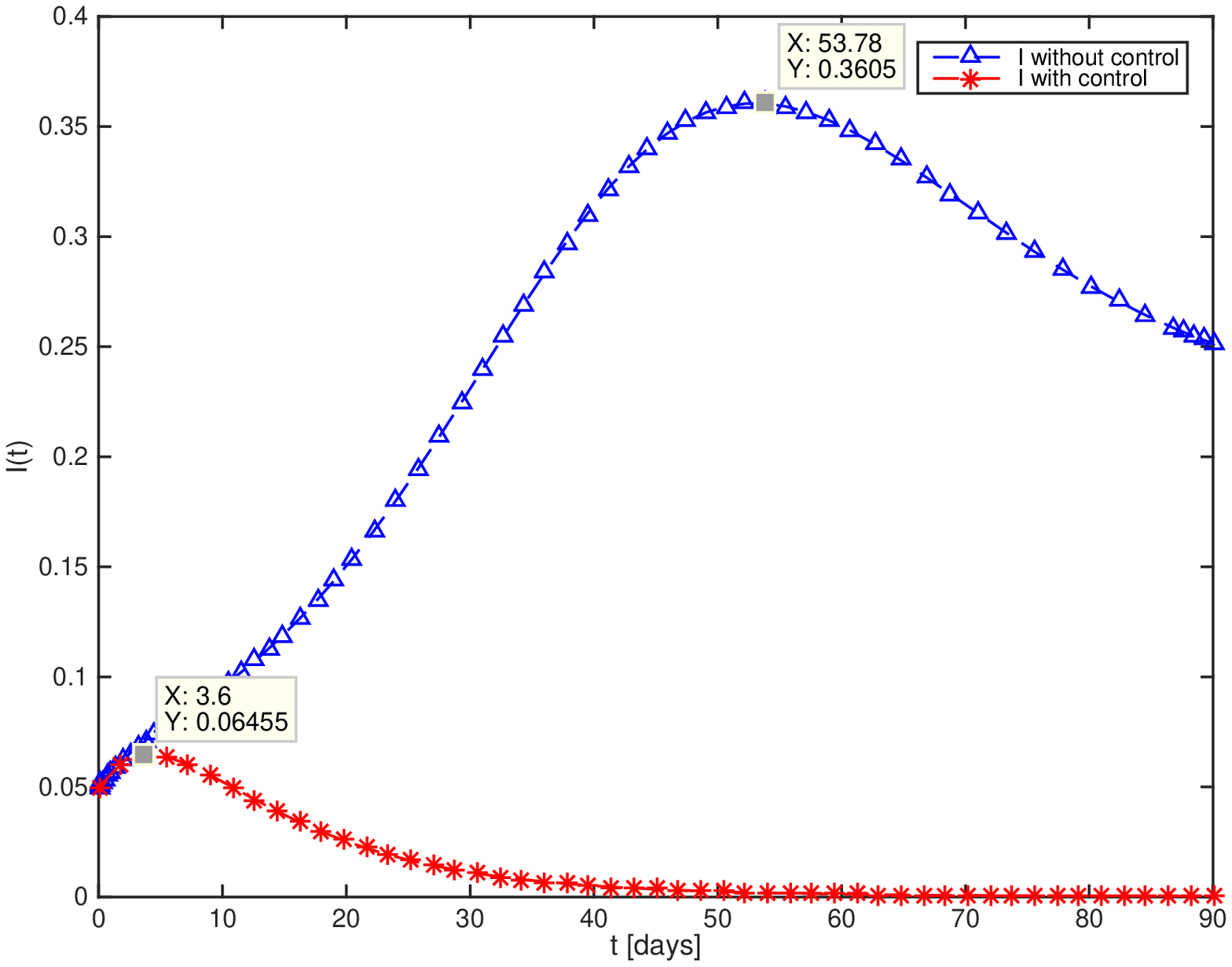}}
\subfloat[Recovered individuals $R(t)$]{
\label{cntrlBD_R}\includegraphics[width=0.5\textwidth]{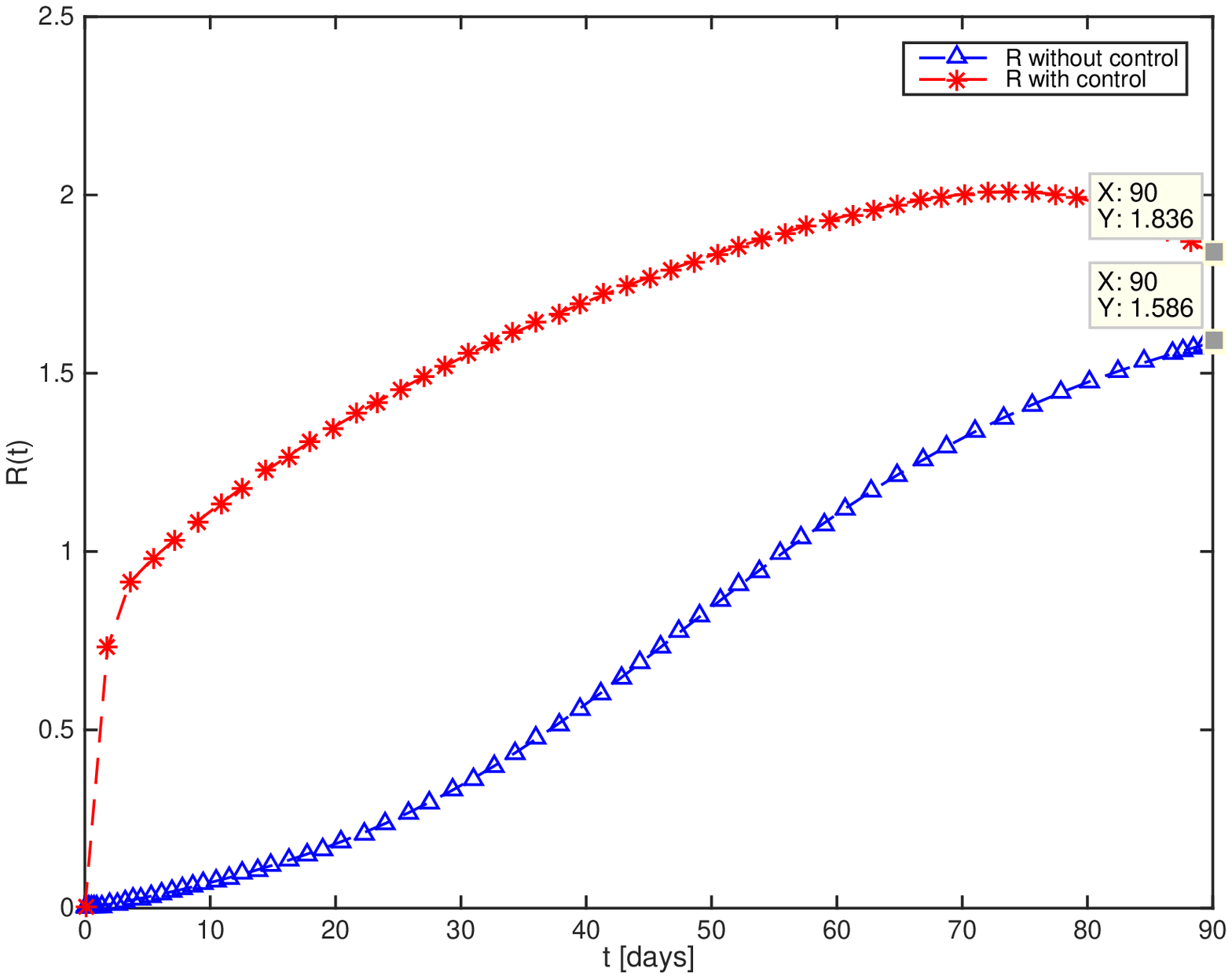}}
\caption{Comparison between the curves of individuals 
$S(t)$, $E(t)$, $I(t)$, $R(t)$, in case of optimal control of
\eqref{SEIR_control}--\eqref{cost_func_strat1} 
\emph{versus} without control \eqref{eq1:SEIR_BD}.}
\label{cntrlBD_SEIR}
\end{figure}
Figure~\ref{cntrlBD_SEIR} shows, respectively, the significant difference 
in the number of susceptible, exposed, infected and recovered individuals, 
with and without control, along time. As expected, the number of susceptible 
individuals $S(t)$ decrease rapidly in case of vaccination 
(see Figure~\ref{cntrlBD_S}), beginning to increase as we decrease vaccination: 
compare Figure~\ref{cntrlBD_S} of $S(t)$ with Figure~\ref{cntrlBD_u}, 
which represents the optimal control function $u(t)$ along time.
\begin{figure}[ht]
\centering
\includegraphics[scale=0.45]{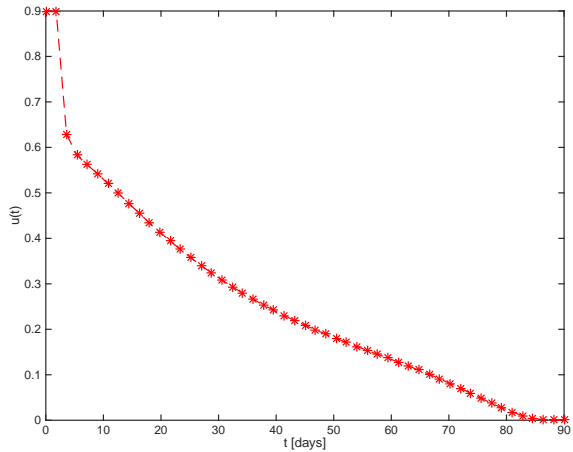}
\caption{The optimal control function $u(t)$
for problem \eqref{SEIR_control}--\eqref{cost_func_strat1}
with initial conditions \eqref{eq:ic}, $t \in [0, t_{end}]$,
$t_{end} = 90$ days, and $\tau = 0.02$.
\label{cntrlBD_u}}
\end{figure}
Figure~\ref{cntrlBD_E} shows that  the number of exposed individuals 
decreases rapidly in case of control. In the same figure, 
the curve of exposed shows that the period of incubation 
of the virus is $22$ days in case of optimal control against more than 
$90$ days  in the absence of any control. In Figure~\ref{cntrlBD_I}, 
the time-dependent curve of infected individuals shows that the peak 
of the curve of infected individuals is less important in presence of control. 
In fact, the maximum value on the infected curve $I$ under optimal control 
is 0.06\%, against 0.36\% without any control (see Figure~\ref{cntrlBD_I}). 
Figure~\ref{cntrlBD_R} shows that the number of recovered individuals
increases rapidly in presence of control. The other important effect of control,
which we can see in Figure~\ref{cntrlBD_R}, is the period of infection, 
which is less important in case of control of the virus. The value 
of the period of infection is $50$ days in case of optimal control, 
in contrast with more than $90$ days without vaccination. 
In conclusion, one can say that Figure~\ref{cntrlBD_SEIR} 
shows the effectiveness of optimal vaccination in controlling Ebola.


\section{Ebola model with vital dynamics and induced death rates}
\label{sec:5}

In this section, we study a second Ebola model with demographic effects 
by increasing the death rates of the exposed and infectious classes 
of the model, that is, by considering induced death rates 
$\lambda_{E}$ and $\lambda_{I}$ associated to the exposed 
and infected individuals, respectively. 


\subsection{Model formulation}
\label{subsec:15}

If we change the previous SEIR model \eqref{SEIR_control} by increasing 
the death rates of the exposed and infectious classes, by adding induced 
death rates $\lambda_{E}$ and $\lambda_{I}$, with $\lambda_{I} > \lambda_{E} > 0$, 
then our SEIR system becomes the one of Figure~\ref{fig2:SEIR2}, 
where we show the relationship between the variables of the new system. 
\begin{figure}[ht]
\centering
\includegraphics[scale=0.4]{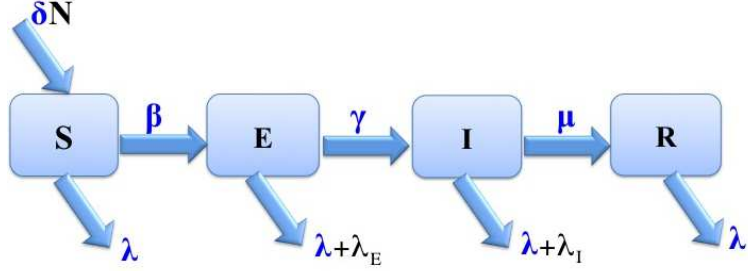}
\caption{Compartment diagram of the SEIR model \eqref{SEIR_BD2}
with induced death rates $\lambda_{E}$ and $\lambda_{I}$
for the exposed and infectious classes, respectively.
\label{fig2:SEIR2}}
\end{figure}
Mathematically, the SEIR model with induced death rates 
is described by the following system of equations:
\begin{equation}
\label{SEIR_BD2}
\begin{cases}
\dfrac{dS(t)}{dt} = \delta N  -\beta S(t)I(t) -\lambda S(t),\\[0.25cm]
\dfrac{dE(t)}{dt} = \beta S(t)I(t) - \gamma E(t)
- \left(\lambda+\lambda_{E}\right) E(t),\\[0.25cm]
\dfrac{dI(t)}{dt} = \gamma E(t) - \mu I(t)
- \left(\lambda+\lambda_{I}\right)  I(t),\\[0.25cm]
\dfrac{dR(t)}{dt} = \mu I(t) - \lambda R(t).
\end{cases}
\end{equation}
Next, we study \eqref{SEIR_BD2}, which we then show to describe in
a very accurate way the recent reality of Ebola in Liberia. 


\subsection{Analysis of the equilibria}
\label{subsec:25}

Let us find the equilibria points of the system of equations
\eqref{SEIR_BD2}. By setting the right-hand side of \eqref{SEIR_BD2}
 to zero, we find that
\begin{equation}
\label{eq1:SEIR_BDeq2}
 \delta N  -\beta SI -\lambda S = 0,
\end{equation}
\begin{equation}
\label{eq2:SEIR_BDeq2}
\beta SI - \gamma E  -\left(\lambda+\lambda_{E}\right) E = 0,
\end{equation}
\begin{equation}
\label{eq3:SEIR_BDeq2}
\gamma E - \mu I  -\left(\lambda+\lambda_{I}\right) I = 0,
\end{equation}
\begin{equation}
\label{eqR:SEIR_BDeq2}
\mu I  -\lambda R = 0.
\end{equation}
By adding \eqref{eq1:SEIR_BDeq2} and \eqref{eq2:SEIR_BDeq2}, we obtain that
\begin{align}
\label{eq4:SEIR_BDeq2}
S = \dfrac{ \delta N -\left(\gamma + \lambda + \lambda_{E}\right)E}{\lambda}.
\end{align}
From \eqref{eq3:SEIR_BDeq2} we obtain that
\begin{align}
\label{eq5:SEIR_BDeq2}
I = \dfrac{\gamma E}{\mu + \lambda + \lambda_{I}},
\end{align}
while from \eqref{eqR:SEIR_BDeq2} it follows that
\begin{align}
\label{eqR1:SEIR_BDeq2}
R = \dfrac{\mu}{\lambda}I.
\end{align}
Using \eqref{eq4:SEIR_BDeq2} and \eqref{eq5:SEIR_BDeq2} in 
\eqref{eq2:SEIR_BDeq2}, one gets
\begin{equation*}
E \left(\dfrac{\beta\left(\delta N -(\gamma +\lambda+ \lambda_{E}) E\right)
\gamma}{\lambda\left(\mu +\lambda+ \lambda_{I}\right)}
- \left(\gamma +\lambda+ \lambda_{E}\right)\right)  = 0,
\end{equation*}
that is, or $E=0$ or
\begin{equation}
\label{eq6:SEIR_BDeq2}
E = \dfrac{\beta\gamma\delta N - \left(\gamma + \lambda
+ \lambda_{E}\right)\left(\mu + \lambda+ \lambda_{I}\right)\lambda}
{\beta \left(\gamma + \lambda+ \lambda_{E}\right)\gamma}.
\end{equation}
If $E=0$, then from \eqref{eq4:SEIR_BDeq2} we obtain $S=(\delta N)/\lambda$,
while from \eqref{eq5:SEIR_BDeq2} $I=0$, which implies by
\eqref{eqR1:SEIR_BDeq2} that $R = 0$. Concluding, 
the virus free equilibrium is
\begin{equation}
\label{eq:DFE2}
P_0=\left(\dfrac{\delta N}{\lambda},0,0,0\right).
\end{equation}
The endemic equilibrium $E^*$ is given by \eqref{eq6:SEIR_BDeq2}, that is,
\begin{align}
\label{eq7:SEIR_BDeq2}
E^* =  \dfrac{\delta N }{\gamma + \lambda+ \lambda_{E}}
- \dfrac{\lambda\left(\mu +\lambda+ \lambda_{I}\right)}{\beta \gamma}.
\end{align}
By using $E$ given by \eqref{eq7:SEIR_BDeq2} in \eqref{eq4:SEIR_BDeq2}, we get
\begin{align}
\label{eq8:SEIR_BDeq2}
S^*=\dfrac{\left(\gamma +\lambda+ \lambda_{E}\right)\left(\mu
+\lambda + \lambda_{I}\right)}{\beta \gamma},
\end{align}
while substituting $E$ given by \eqref{eq7:SEIR_BDeq2}
into \eqref{eq5:SEIR_BDeq2}, we obtain that
\begin{align}
\label{eq9:SEIR_BDeq2}
I^* = \dfrac{\gamma \delta N }{\left(\mu +\lambda
+ \lambda_{I}\right)\left(\gamma +\lambda+ \lambda_{E}\right)}
 - \dfrac{\lambda}{\beta}.
\end{align}
Finally, using \eqref{eqR1:SEIR_BDeq2}, one gets
\begin{align}
\label{eq9:SEIR_BDeq2R}
R^* = \dfrac{\mu}{\lambda}I^* = \dfrac{\mu}{\lambda}\left[
\dfrac{\gamma \delta N }{\left(\mu +\lambda
+ \lambda_{I}\right)\left(\gamma +\lambda+ \lambda_{E}\right)}
- \dfrac{\lambda}{\beta}\right].
\end{align}
The  equilibrium point is then $P^*=\left(S^*,E^*,I^*,R^*\right)$
with expressions for $S^*$, $E^*$, $I^*$ and $R^*$ given by 
\eqref{eq7:SEIR_BDeq2}--\eqref{eq9:SEIR_BDeq2R}. 

\begin{theorem}
\label{thm2}
Let $S(t)$, $E(t)$, $I(t)$, $R(t)$  be a solution of the SEIR model 
\eqref{SEIR_BD2}. Then the basic reproduction ratio is given by 
\begin{align}
\label{eq:R0:2}
R_0 := \dfrac{\beta \gamma \delta N}{\left(\mu +\lambda
+ \lambda_{I}\right)\left(\gamma +\lambda+ \lambda_{E}\right)\lambda}.
\end{align}
\begin{itemize}
\item If $R_0>1$, then the  equilibrium $P^*=\left(S^*,E^*,I^*,R^*\right)$ 
of the virus is obtained, in agreement with expressions
\eqref{eq7:SEIR_BDeq2}--\eqref{eq9:SEIR_BDeq2R}. 
In this case, the virus is able to invade the population.

\item If $R_0<1$, then the disease free equilibrium $P_0$ \eqref{eq:DFE2} 
is obtained. It corresponds to the case when the virus dies out (no epidemic).
\end{itemize}
\end{theorem}

\begin{proof}
We apply the next generation method to the model \eqref{SEIR_BD2}. Here, 
\[
F = \left( \begin{array}{cc}
0 & \dfrac{\beta N \delta}{\lambda} \\ [0.25cm]
0 & 0 \end{array} \right),
\] 
where $\delta$ is the birth rate and $\lambda$ is the death rate of susceptible,
and 
\[
V =\left( \begin{array}{cc}
\gamma + \lambda +\lambda_{E} & 0 \\ [0.25cm]
-\gamma & \mu + \lambda  +\lambda_{I}\end{array} \right).
\] 
Then, $R_0$ is the dominant eigenvalue of $FV^{-1}$, 
which is given by \eqref{eq:R0:2}.
\end{proof}

Note that if $\lambda_I = \lambda_E = 0$, then Theorem~\ref{thm2} 
reduces to Theorem~\ref{thm1}.


\subsection{SEIR model with demographic effects and induced 
death rates, and Liberia's 2014 Ebola outbreak}
\label{subsec:35}

Now, we present a modeling study of the real outbreak of Ebola virus occurred 
in Liberia in 2014 by using World Health Organization (WHO) data. 
The epidemic data used in this study is available at \cite{WHO:Lib:2014}.
Let us start by the analysis of the parameters of the SEIR model  
with demographic effects and induced death rates. 
The birth rate $\delta=0.03507$ and death rate $\lambda=0.0099$ 
of the model are obtained from the specific statistical data 
of the demography of Liberia, available at \cite{indexmundi}.  
To estimate the parameters $\beta$, $\gamma$, $\mu$, $\lambda_{E}$ and $\lambda_{I}$, 
we adapted the initialization of $I$  with the reported data of WHO by
fitting the real data of confirmed cases of infectious in Liberia. 
The result of fitting is shown in Figure~\ref{I_sml_data_model_induced}. 
\begin{figure}[ht]
\centering
\includegraphics[scale=0.45]{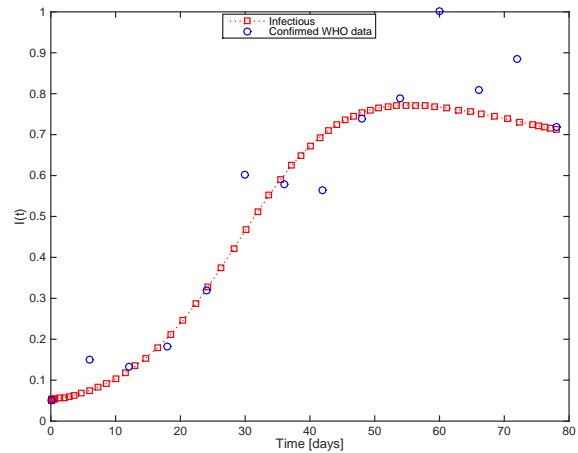}
\caption{Graph of infected  obtained from \eqref{SEIR_BD2} 
and \eqref{eq:ic_induced} versus the
real data of confirmed cases for the 2014 Ebola outbreak occurred in
Liberia \cite{WHO:Lib:2014}.\label{I_sml_data_model_induced}}
\end{figure}
The comparison between the curve of infectious obtained by our simulation 
and the reported data of confirmed cases by WHO shows that
the mathematical model \eqref{SEIR_BD2} fits well the real data by using
$\beta=0.299$ as the rate of transmission, $\gamma=0.034$ as the rate of infectious, 
$\mu=0.0859$ as the recovery rate, and $\lambda_{E}=0.0003366$ and  
$\lambda_{I}=0.0031$ as induced death rates. 
To measure the goodness of fit, we have used a deterministic 
approach for the estimation of the parameters.
Precisely, our fitting procedure has used a 
least squares method of the nonlinear system of ordinary 
differential equations that describes the model.
According to the definition of the least squares method, 
the best-fit curve is the one that provides a minimal 
squared sum of deviation from real data. In our case,
the fitting procedure is associated with the numerical resolution 
of the nonlinear system of ordinary differential equations \eqref{SEIR_BD2}
that describes the model. Accurately, the parameters were estimated by solving
the following nonlinear programming optimization problem:
\begin{eqnarray}
\label{optim}
\begin{array}{ll}
\mbox{minimize} & \mathcal{D} = 
\displaystyle\sum_{j=1}^{n} 
\left(I_{\text{real}, j} - I_{\text{siml}, j}\right)^2 \\
\mbox{subject to} & \mbox{equations of the model} \,  \eqref{SEIR_BD2},
\end{array}
\end{eqnarray}
where  $I_{\text{real}, j}$ corresponds to real data 
and $I_{\text{siml}, j}$ to the one obtained from 
the resolution of the nonlinear system of ordinary equations \eqref{SEIR_BD2}. 
The goodness of fit is measured by computing the value of the objective 
function $\mathcal{D}$ of \eqref{optim}, which is in our case
equal to $0.0002$. By comparing the values of $\lambda_{E}$ 
and $\lambda_{I}$ with the value of the death rate $\lambda$, 
we remark that $\lambda_{E} = 0.0340 \lambda =3.4 \% \lambda$ 
and $\lambda_{I} = 0.3131 \lambda =31.31 \% \lambda$. 
By using the value of Liberia's population, which is estimated 
at $P=4.4$ million in $2014$, and  the number of confirmed 
infectious cases obtained from WHO, which is given by $d=230000$, 
we fix $I(0)=\dfrac{d}{P}$, meaning that the confirmed number 
of infectious cases represents $5.23\%$ of the total population. 
The initial susceptible, exposed, infectious and recovered populations, 
are given respectively by
\begin{equation}
\label{eq:ic_induced}
S(0)=0.8977, \quad E(0)=0.05, \quad I(0)=0.0523, \quad R(0)=0. 
\end{equation}
Figure~\ref{I_sml_data_model_induced} gives 
the curve of infectious individuals simulated 
with \eqref{SEIR_BD2} subject to \eqref{eq:ic_induced} 
and the one obtained from the WHO real data. Note that the choice of $I(0)$ 
in \eqref{eq:ic_induced} is in agreement with WHO's data shown in 
Figure~\ref{I_sml_data_model_induced}. 
The evolution of all groups of individuals, over
time, is shown in Figure~\ref{fig:BD:SEIR2}.  
\begin{figure}[ht]
\centering
\subfloat[Evolution of $S(t)$ with $S(0)=0.8977$]{
\label{BD_fig1:SEIR2_fit}\includegraphics[width=0.5\textwidth]{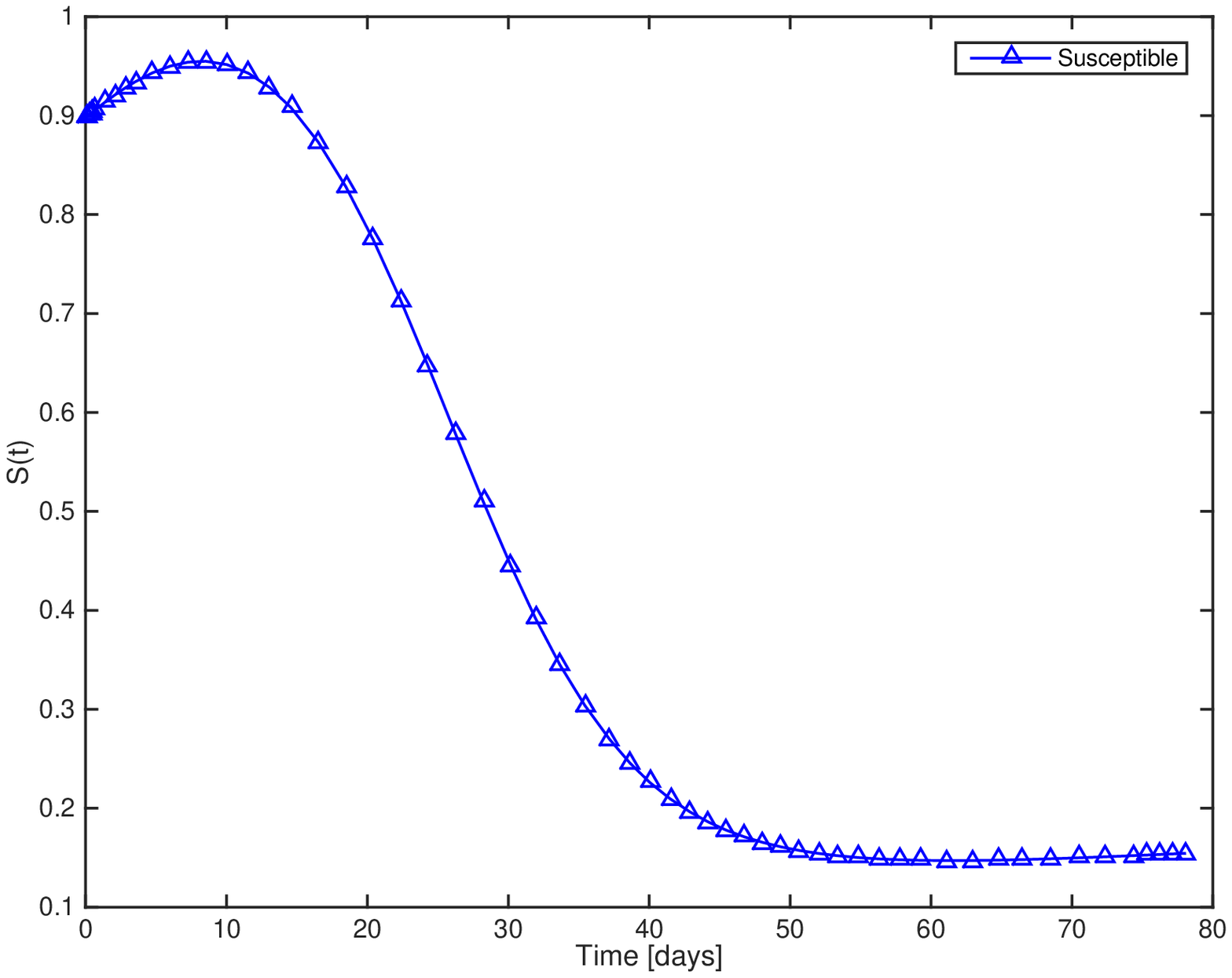}}
\subfloat[Evolution of $E(t)$ with $E(0) = 0.05$]{
\label{BD_fig2:SEIR2_fit}\includegraphics[width=0.5\textwidth]{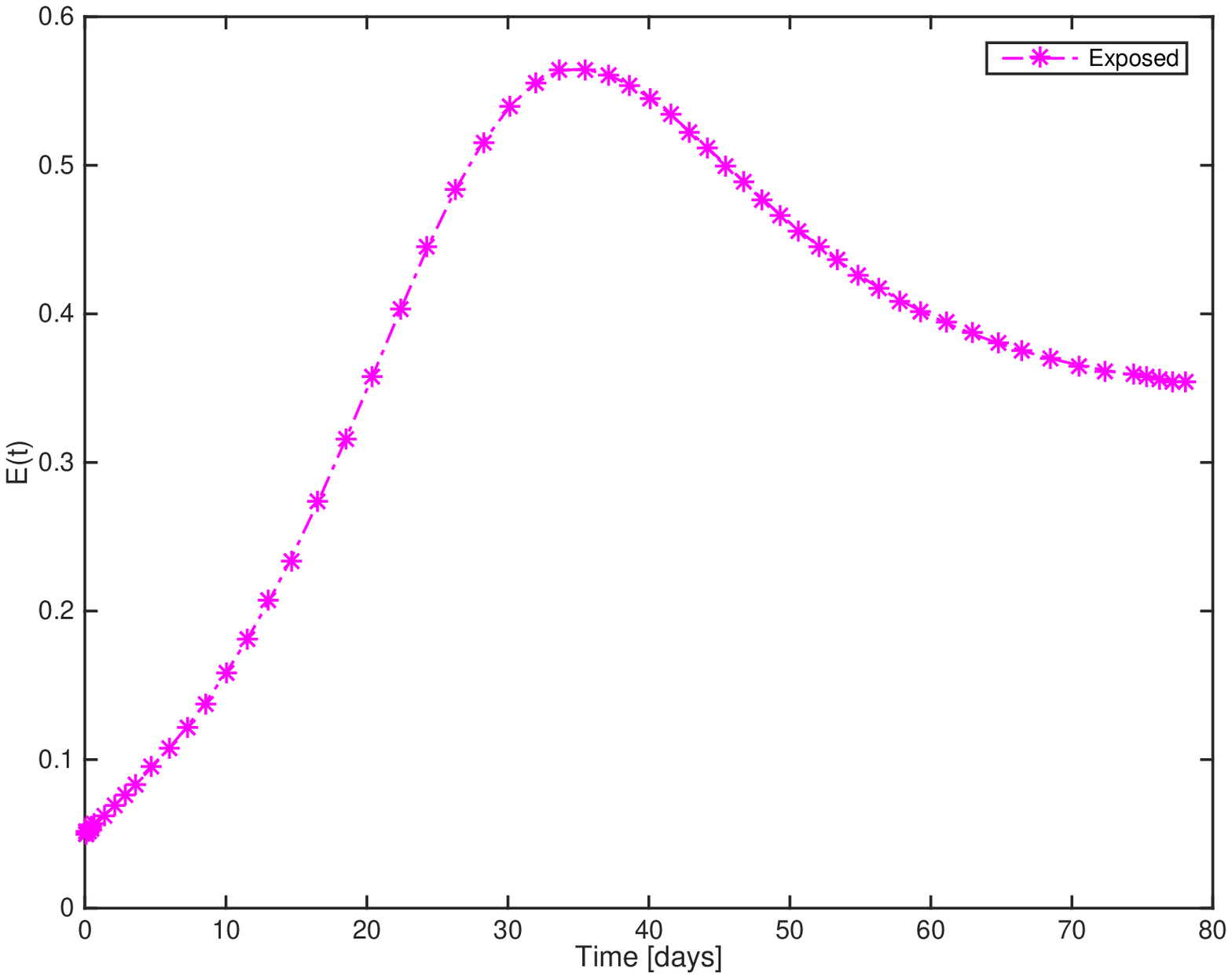}}\\
\subfloat[Evolution of $I(t)$ with $I(0) = 0.0523$]{
\label{BD_fig3:SEIR2_fit}\includegraphics[width=0.5\textwidth]{I_identf.eps}}
\subfloat[Evolution of $R(t)$ with $R(0) = 0$]{
\label{BD_fig4:SEIR2_fit}\includegraphics[width=0.5\textwidth]{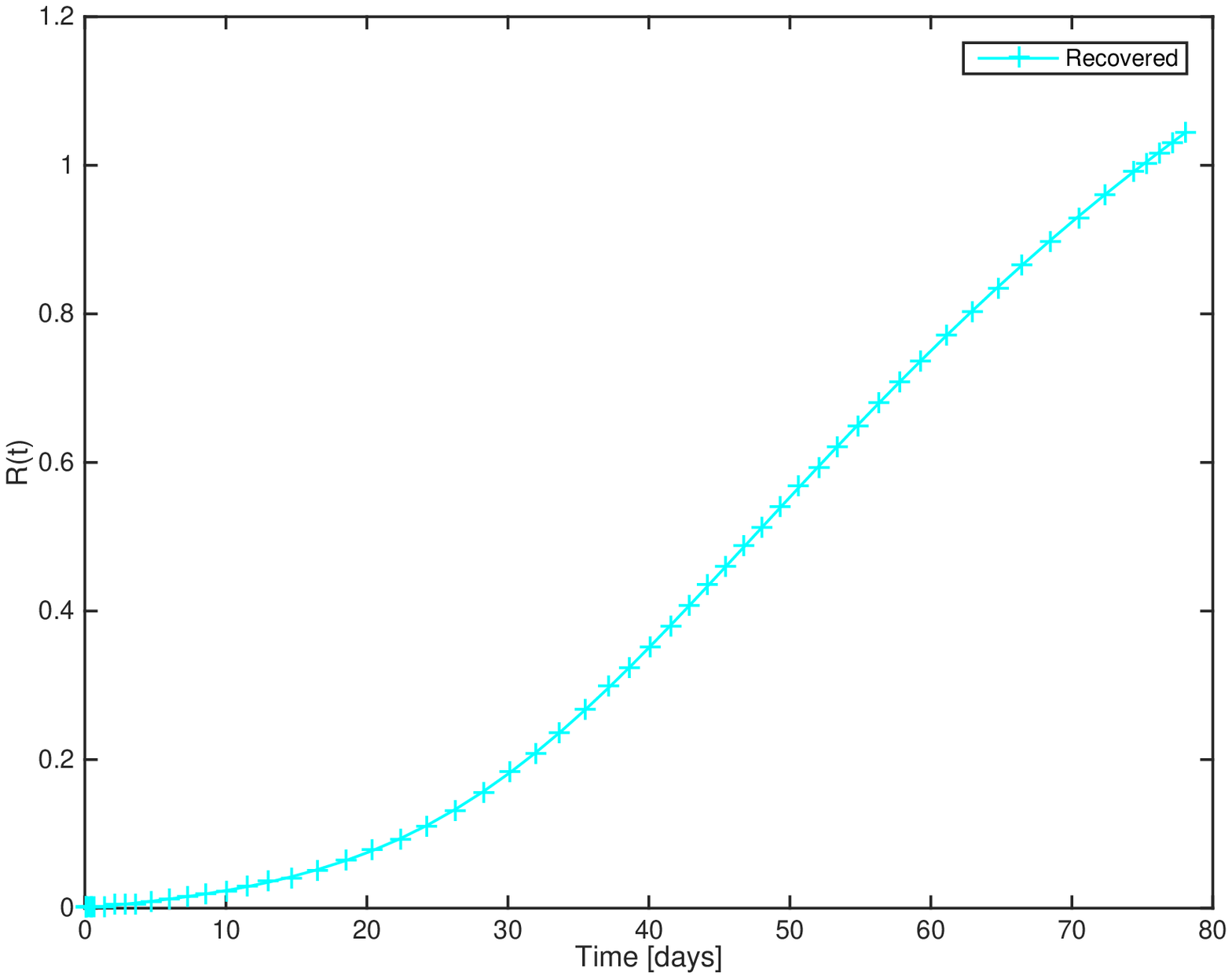}}
\caption{Evolution of individuals in compartments
$S(t)$, $E(t)$, $I(t)$ and $R(t)$ of the SEIR model \eqref{SEIR_BD2} 
with vital dynamics and induced death rates, where the
initial numbers of susceptible, exposed, infectious,
and recovered groups and the parameter values of the model
are described in Section~\ref{subsec:35}.}
\label{fig:BD:SEIR2}
\end{figure}


\section{Conclusion}
\label{sec:conc}

We investigated several SEIR models in the context of the recent 
Ebola outbreak in West Africa. Our aim was to study the properties
and usefulness of SEIR models with respect to Ebola. We began by 
presenting the basic SEIR model and its mathematical analysis. 
Then, we added to the model demographic effects in order to analyze 
the equilibria with vital dynamics. The system of equations 
of the model was solved numerically. The numerical simulations 
confirm the theoretical analysis of the equilibria of the model. 
Moreover, using optimal control, we controlled the propagation 
of the virus through vaccination, reducing 
the number of infected individuals and taking into account 
the cost of vaccination. Finally, we considered a more complete model
with induced death rates and have shown its usefulness with respect
to Liberia's outbreak of 2014.


\section*{Acknowledgements}

This research was partially supported by the Portuguese Foundation
for Science and Technology (FCT): through project UID/MAT/04106/2013
of the Center for Research and Development in Mathematics 
and Applications (CIDMA); and within project TOCCATA, 
Ref. PTDC/EEI-AUT/2933/2014, co-funded by Project 
3599, Promover a Produ\c{c}\~ao Cient\'{\i}fica e Desenvolvimento
Tecnol\'ogico e a Constitui\c{c}\~ao de Redes Tem\'aticas (3599-PPCDT),
and FEDER funds through COMPETE 2020, Programa Operacional
Competitividade e Internacionaliza\c{c}\~ao (POCI).
The authors are grateful to two referees for valuable comments 
and suggestions, which helped them to improve the quality of the paper. 



\end{document}